\documentclass[aps,twocolumn,nofootinbib,preprintnumbers,superscriptaddress]{revtex4-2}

\usepackage{color}
\usepackage{bbm}
\usepackage{tikz}

\usepackage[
pagebackref=false,
colorlinks=true,
linkcolor=blue,
urlcolor=blue,
filecolor=black,
citecolor=red,
pdfstartview=FitV,
pdftitle={},
pdfauthor={},
pdfsubject={},
pdfkeywords={},
pdfpagemode=None,
bookmarksopen=true
]{hyperref}

\usepackage{amsmath,amssymb}
\usepackage{amsthm}

\usepackage{physics}
\usepackage{graphicx}
\usepackage{bm}
\usepackage{bbm}
\usepackage{color} %for colored comments
\usepackage{srcltx} 
\usepackage{newtxtext,newtxmath}
\usepackage{cleveref}

\allowdisplaybreaks
\usepackage{quantikz}
\usepackage{xcolor}
\newtheorem{theorem}{Theorem}[section]
\newtheorem{corollary}{Corollary}[theorem]
\newtheorem{lemma}[theorem]{Lemma}
\usepackage{comment}
\usepackage{varwidth}

\newcommand{\nb}[1]{\color{blue}}

\newcommand{\hl}[1]{\color{magenta}}

\begin{document}
\title{Canonical Partition Function on a Quantum Computer through Trotter Interpolation}
\date{\today}

\author{Taozhi Guo}
\thanks{taozhig@princeton.edu. This work was completed while the first author was an intern at Fujitsu Research of America, Inc.}
\affiliation{Department of Physics, Princeton University, Princeton, New Jersey 08544, USA}

\author{Gumaro Rendon} 
\thanks{grendon@fujitsu.com}
\affiliation{Fujitsu Research of America, Inc, Santa Clara, CA 95054, USA}

\author{Rutuja Kshirsagar}
\thanks{rkshirsagar@fujitsu.com}

\affiliation{Fujitsu Research of America, Inc, Santa Clara, CA 95054, USA}

\begin{abstract}

  In this work, we present a Gibbs state observable estimation algorithm based on Trotter interpolation, which reaches a state-of-the-art quantum computational cost of $ \tilde{O}(\beta \log{1/\epsilon})$. Our approach saves $\log(\Gamma)$ ancilla qubits compared with the qubitization-based methods for Hamiltonian with $\Gamma$ stages. To provide a robust assessment of our approach, we benchmark our results against state-of-the-art methodology using the SYK model as a testbed. Our method provides an efficient alternative method for Gibbs-state accessing based on Trotterization in the context of quantum state preparation and estimation of thermal observables.
\end{abstract}
\maketitle

\tableofcontents

\section{Introduction}

Traditionally, classical Gibbs sampling or a Gibbs sampler is a Markov chain Monte Carlo (MCMC) algorithm \cite{Duane:1987de, Kennedy1990} used as a means of statistical inference. Recently, it has also been used as an efficient technique for sampling the posterior distribution of neural network weights in Bayesian deep learning by adding noise at each activation layer\cite{Piccioli_2024}, which offers a promising alternative for deep learning inference.
Quantum Gibbs states (also known as thermal states) describes quantum state in a finite temperature and it can be used for a variety of cases such as the study of phase transitions in many-body quantum systems, in quantum machine learning ~\cite{Kieferova2017, Biamonte2017}, quantum optimization~\cite{Somma2008}, in Semi-definite programming~\cite{SDP_Brandao_2016,van_Apeldoorn_2020} and also in the study of open quantum systems~\cite{Poulin2009}. Particularly, the Gibbs state is important for quantum many-body physics as the finite temperature properties of many such models are not well understood. 

\begin{comment}
\textbf{GR:If we change the qubit savings estimate to something like the electronic structure Hamiltonian we need to change this paragraph} One of the most widely studied systems in quantum many-body physics and material science is the Fermi-Hubbard model\cite{Arovas_2022}, which describes interacting Fermions on a lattice. 
Such a model has been helpful in understanding exotic phenomena such as superconductivity\cite{RevModPhys.78.17,osti_1357580}, metal-insulator transition\cite{RevModPhys.70.1039}, and spin liquid\cite{Savary_2016, RevModPhys.89.025003}. 
Studying the Gibbs state of such a model is particularly important as the finite temperature behavior of the model is not well understood due to the difficulties in numerical and experimental analysis. 
Exponential growth of the Hilbert space and strong interactions result in high computational costs. 
With the help of quantum computers, the simulation of finite temperature behaviors of the Fermi-Hubbard model could potentially go beyond the ability of current classical studies. 
\end{comment}

One of the most important many-body physics models is the SYK model \cite{kitaev_15,PhysRevLett.70.3339} which describes the random $q$ couplings between $n_{\rm Majorana}$ Majorana Fermions. The model has not only attracted significant attention in the field of condensed matter physics, such as high temperature superconductors \cite{Patel_2023, Sachdev_2010,chowdhury2022sachdev}, but also in high energy physics, particularly in holography \cite{Maldacena_2016,Cotler_2017} as the SYK model has a holographic dual. A feature of the SYK model is the maximally quantum chaotic as large $n_{\rm Majorana}$ limit \cite{nosaka2020quantumchaosthermodynamicsblack, Cotler_2017} which serves as an important tool to understand the chaotic behavior of black holes. As a toy model for quantum gravity, it is crucial to study the dynamical and thermal properties of the model. 
So far, it has been challenging to study the model as it is not analytically solvable in most cases. 
Thus, quantum computers could provide an efficient method for simulating the SYK model at finite temperature with Gibbs states.

Although there has been a demonstrated quantum speed-up over general Gibbs state preparation/sampling techniques, this is limited to a quadratic speed-up rooted in the energy-time uncertainty limit. Some approaches are inspired by nature in that they use an ancillary system \cite{Poulin2009, chowdhury2016quantumalgorithmsgibbssampling} or a master equation (Linbladian)\cite{chen2023quantum,chen2023efficientexactnoncommutativequantum} to simulate the thermalization by contact with a heat bath/sink. They have proven polynomial complexity bounds with respect to the mixing time. However, the mixing time is polynomial with respect to the system size or polylog in the Hilbert space size for a general system has not been demonstrated ~\cite{Poulin2009}. 
Ref.~\cite{chen2023quantum} introduces a method using continuous-time quantum Gibbs samplers to efficiently prepare thermal states. By leveraging Lindbladian dynamics, the authors prove finite-time thermalization and construct detailed-balanced processes for general quantum systems. This work offers a quantum analog to classical Monte Carlo methods, advancing scalable quantum simulations. The hope for these approaches is that for physically inspired systems, this time is polynomial in size as demonstrated by nature. Further shown in Ref.~\cite{chen2023efficientexactnoncommutativequantum}, by leveraging quantum phase estimation, the method provides accurate preparation of Gibbs states, addressing common thermalization issues and advancing quantum thermal simulations.
Ref.~\cite{chowdhury2016quantumalgorithmsgibbssampling} explores the role of quantum computing in accelerating Gibbs sampling and hitting-time estimation and shows that quantum phase estimation can achieve quadratic speedups in preparing Gibbs states and estimating hitting times, potentially revolutionizing MCMC convergence and enhancing quantum simulations.

Some other method for thermal Gibbs sampling includes Quantum annealing \cite{PhysRevApplied.17.044046}, which is adapted to generate samples from Gibbs distributions, while addressing challenges like noise and imperfections in quantum hardware. 
Combining with phase estimation \cite{Bergamaschi_2024}, such an approach offers polynomial-time solutions to problems that would otherwise require exponential time for classical computation. 
The exponential advantage over classical methods could be extended to to "O(1)-local" Hamiltonians with limited particle interactions \cite{rajakumar2024gibbssamplinggivesquantum}. 
The approach of truncating the cluster expansion of the Gibbs state \cite{Eassa_2024, 10.1063/1.1705066, Oitmaa_Hamer_Zheng_2006} could also provide a good approximation of the Gibbs state while reducing computational complexity and improving sampling accuracy in systems with short-range interactions. 
Quantum Singular Value Transformation (QSVT)\cite{Gily_n_2019} could also be used to prepare Gibbs states by leveraging techniques of block-encoding and quantum signal processing. There have been some other efforts for implementing the imaginary time evolution operator through a semi-classical linear combination of unitaries expansion\cite{matsumoto2024quantummanybodysimulationfinitetemperature} for finite temperature simulation, this with lower quantum overheads than the methods that rely on quantum-signal processing which promise to be implementable on earlier quantum devices, but at the cost of losing the famous quadratic advantage. 

In this work, we focus on providing an alternative to state-of-the-art methods that rely on block-encoding of the Hamiltonian, for example, \cite{Gily_n_2019}, which work for a general Hamiltonian. Our contribution provides a more modest quadratic improvement for estimating Gibbs state observables with a quantum computational cost of $  \tilde{O}(\log{1/\epsilon})$ based on Chebyshev interpolation for the Trotterization, which is proper of these types of methods. \textit{Futher, we show that compared with qubitization-based methods, our approach achieves ancillary qubits saving of $\log{\Gamma}$ qubits by bypassing the usage of the quantum walk operator}. 

The rest of the paper is organized as follows: Section \ref{sec:alg_summary} provides an overview of our algorithm. In Section \ref{sec:estimates}, we outline the fundamentals of our approximations and estimate the required polynomial size for achieving targeted error thresholds. Section \ref{sec:QSP_for_Trotterized_Ev_op} demonstrates the application of generalized quantum signal processing techniques from \cite{motlagh2024generalizedquantumsignalprocessing} to our Trotterized unitary operator. Section \ref{sec:cost_analysis} offers a cost analysis of the various steps involved in our algorithm. In Section \ref{sec:numerical_results}, we apply our analysis to the SYK model and show the auxiliary qubits saved by our method. Finally, Section \ref{sec:conclusions} presents the conclusions of our work and future directions.

\section{The Algorithm\label{sec:alg_summary}}

We want to estimate the canonical partition function
\begin{align}
    \frac{Z(\beta)}{N} = \frac{{\rm Tr} \left(e^{-\beta H}\right)}{N},
\end{align}
where the Hamiltonian $H$ on a system of $n$ qubits ($N=2^n$), can be decomposed into a sum of $\Gamma$ terms:
\begin{align}
    H = \sum^{\Gamma}_{\gamma=1} H_{\gamma}
\end{align}
Notice that both the Hamiltonian and decomposition are specified, as the decomposition is not unique, but the assumption is that each individual $H_{\gamma}$ is fast-forwardable, like, for example, any term proportional to a Pauli-string. Throughout this work, $\sum_\gamma^{\Gamma}\|H_{\gamma}\|\leq 1$ is assumed.

The algorithm derivation can be summarized in the following steps:
\begin{itemize}
    \item We approximate the Hamiltonian evolution through a $p$th-order product formula which corresponds to the exact evolution operator for an effective Hamiltonian $\tilde{H}_{p}(s t)$ for some Trotter step size $(s t)$.
    \item  From these evolution operators, we construct an implementation of the operator $e^{-i \beta \tilde{H}_p (s t)}$, for an inverse temperature $\beta$.
    \item For a set of Trotter step sizes $ t s_k $, we estimate the trace $Z(\beta,s_k)={ \rm Tr} \left(e^{-i \beta \tilde{H}_p (s_k t)}\right)$ on a quantum device.
    \item From these estimates, we use Chebyshev interpolation to obtain the partition function at $s=0$.
\end{itemize}

Standard examples of product formulas for other Hamiltonian $H = \sum_{\gamma = 1}^{\Gamma} H_{\gamma}$ include the Trotter first and second order formula following Ref.~\cite{10.1063/1.529425}. 
\begin{align}
S_{1} (t) := \prod^{\Gamma}_{\gamma=1} e^{i H_{\gamma} t}, \quad S_{2} (t) :=  S_1 (t/2) S_1(-t/2)^{\dagger}
\end{align}
and more generally, the order $2l$ symmetric Suzuki-Trotter formula, defined recursively as
\begin{align} \label{eq:symm_S2k}
S_{2l} (t) := [S_{2l-2}(u_l t)]^2 S_{2l-2,j}\left((1 - 4 u_l) t\right) [S_{2l-2}(u_l t)]^2
\end{align}
for every $l\in \mathbb{Z}_+\setminus\{1\}$, where $u_l := (4-4^{1/(2l-1)})^{-1}$. 

From these, one can define an effective Hamiltonian as 
\begin{align}
    \tilde{H}_{p}(s t) = \frac{\log\left(S_{p}(s t)\right)}{i s t },
\end{align}
where $\tilde{H}_{p}(s t)$ will be a single-valued analytic function of $s$ provided that $t \leq \pi$.

Our results can be summarized in the following short version of our main results theorem:
\begin{theorem}[Short Version]\label{thm:main_informal}
There exists an algorithm to estimate the canonical partition function for an n-qubit system with a Hamiltonian $H=\sum^{\Gamma}_\gamma H_\gamma$ for fast-forwardable $H_{\gamma}$ and $\| H \| \leq 1$, which has a quantum computational cost of:
\begin{align*}
     \tilde{O}\left(\beta \log{1/\epsilon}\left(\frac{\sqrt{\max_{s} Z(\beta,s)/2^n}}{\varepsilon}\right)\right),
\end{align*}
where $\epsilon$ is the \textbf{algorithmic error} and $\varepsilon$ is the \textbf{statistical error}. This, with the use of 2n+2 qubits and without the need for block-encoding ancillary qubits.
\end{theorem}
Here $n$ ancillary qubits are for the estimation of the trace~\cite{chowdhury2016quantumalgorithmsgibbssampling}, and the other 1 ancillary qubit is for the Generalized Quantum Signal Processing (GQSP) we use from ~\cite{motlagh2024generalizedquantumsignalprocessing}, and the other 1 ancillary qubit is used for quantum amplitude estimation ( See for example IQAE in \cite{grinko2021iterative}).
In the next section, we will introduce the tools we use for the polynomial expansion of $e^{\beta(x+1)}$ needed for GQSP and the Chebyshev interpolation methods we use for interpolating towards $s\to 0$.

\section{Polynomial Size Estimates } \label{sec:estimates}
In this section, we provide the basics of the approximations and estimate the polynomial size needed for targeted errors. Both for the polynomial expansions needed for the signal processing implementing $e^{-\beta(x+1)}$, and for the Chebyshev interpolation needed for interpolation of $Z(\beta,s_k)/N$ towards $s\to 0$. 

\subsection{Low weight Fourier series approximation}
To implement the function $e^{-\beta(x+1)}$, we consider the low-weight Fourier series approximation provided in Ref.~\cite{van_Apeldoorn_2020}. The Lemma states that  
\begin{lemma}\label{lemma:LowWeightAPX}
  Let $\delta,\varepsilon\in\!(0,1)$ and $f:\mathbb{R}\rightarrow \mathbb{C}$ s.t. $\left|f(x)\!-\!\sum_{k=0}^K a_k x^k\right|\leq \varepsilon/4$ for all $x\in\![-1+\delta,1-\delta]$.
  Then $\exists\, c\in\mathbb{C}^{2M+1}$ such that
  $$
  \left|f(x)-\sum_{m=-M}^M c_m e^{\frac{i\pi m}{2}x}\right|\leq \varepsilon
  $$
  for all $x\in\![-1+\delta,1-\delta]$, where $M=\max\left(2\left\lceil \ln\left(\frac{4||a||_1}{\varepsilon}\right)\frac{1}{\delta} \right\rceil,0\right)$ and $||c||_1\leq ||a||_1$. Moreover, $c_m$ can be efficiently calculated on a classical computer in time $\text{poly}(K,M,\log(1/\varepsilon))$.
\end{lemma}

From the above Lemma.~\ref{lemma:LowWeightAPX}, we could derive the following corollary for the required order of the low-weight Fourier expansion for the Gibbs function $e^{-\beta (x+1)}$. 
\begin{corollary} \label{cor:M_k}
    For the function of $e^{-\beta(x+1)}$ with $x\in\![-1+\delta,1-\delta]$, the required order of expansion with error $\varepsilon$ scales as $M \sim O(\beta \ln(1/\varepsilon))$
\end{corollary}

For Hamiltonian $H$ with norm $||H|| \leq 1 $, we shift the Boltzmann function to $e^{-\beta (x+1)}$ to  to make the argument $H+1$ positive definite and keep the norm of $||e^{-\beta (H+1)}|| \leq 1$. 
Observe that when the ground state energy of $H$ lies at exactly $-1$, the polynomial degree needed to approximate the function becomes infinite as shown in Lemma.~\ref{lemma:LowWeightAPX}. 
Therefore, by identifying $x$ with $\frac{H + \delta}{1 + \delta}$, the ground state energy of $x$ is away from $-1$ by $O(\delta)$. 
Further, to avoid the exponential sub-normalization of the $e^{-\beta \delta}$, we require that $\delta \sim 1/\beta$.

\subsection{ Chebyshev expansion} % required for implementing $e^{-\beta (H+1)}$ through direct Chebyshev expansion or Jacobi-Anger expansion of previous trigonometric expansion}

\begin{comment}
The Chebyshev expansion of the function could be expressed as
\begin{equation}
    e^{-\beta (x+1)} = \sum_{n = 0}^{\infty} a_n T_n(x)
\end{equation}
where $T_n(\cos(\theta)) = \cos(n\theta)$ are degree - n Chebyshev polynominal and $a_n$ are Chebyshev coefficients
\begin{equation}
    a_n = \frac{2}{\pi} \int_{-1}^1 \frac{f(x)T_n(x)}{\sqrt{1 - x^2}}dx
\end{equation}
As could be shown, for an analytical function $f$ in $[-1, 1 ]$ and analytically continuable to the interior of the Bernstein ellipse $E_p = {\frac{1}{2} (z+ z^{-1}) : |z| = \rho}$ where $|f(x)| \leq M$, then the trunction error of the Chebyshev expansion satisfies
\begin{equation}
   \varepsilon =  ||f - f_{n}||_{[-1,1]} \leq \frac{2M \rho^{-n}}{\rho - 1}
\end{equation}
Hence, for $\rho = 1 + \frac{1}{\beta} $, $|e^{-\beta(x+1)}| \leq 1$ we get 

\begin{equation}
    \varepsilon \leq 2 \beta  (1 + \frac{1}{\beta})^{-n}
\end{equation}

\end{comment}

To interpolate $Z(\beta,s_k)/N$ towards $s\to 0$, we follow the methods of \cite{Rendon2024improvedaccuracy,rendon2023dequantizing}. We set 
\begin{align}
f(s)= Z(\beta,s)/N
\end{align}
Let $P_{M_{\rm cheb}-1}f(s) = \sum_{j = 0}^{M_{\rm cheb}-1}a_j u_j(s)$ be the $M_{\rm cheb}-1$ linear order interpolate, then we could write the interpolate by solving the following linear equation for the vector $f = [f(s_1), f(s_2),...,f(s_{M_{\rm cheb}})]$ and $a = [a_1, a_2,...,a_{M_{\rm cheb}}]$ :
\begin{align}
    f = \mathbf{V}  a ,
\end{align}
where
\begin{align}
 \mathbf{V} :=
\begin{pmatrix}
 u_0(s_1)   & u_1(s_1)   & \dots  & u_{{M_{\rm cheb}}-1}(s_1) \\
 u_0(s_2)   & u_1(s_2)   & \dots  & u_{{M_{\rm cheb}}-1}(s_2) \\
 \vdots     & \vdots     & \ddots & \vdots  \\
 u_0(s_{M_{\rm cheb}}) & u_1(s_{M_{\rm cheb}}) & \dots  & u_{{M_{\rm cheb}}-1}(s_{{M_{\rm cheb}}})
\end{pmatrix}.
\end{align}
With this, one can obtain the coefficients, $a_j$, with $a=\mathbf{V}^{-1}f$. The choice of interpolating nodes $s_k$ and the interpolating set of polynomials are the Chebyshev nodes and polynomials.
In this case, $u_j$ is defined by
\begin{align} \label{eq:cheb_orthonorm_def}
     u_j(s) :=
    \begin{cases}
    \sqrt{\frac{1}{M_{\rm cheb}} }T_0(s), &j=0 \\
    \sqrt{\frac{2}{M_{\rm cheb}}}T_j(s), &j=1,2,\dots,M_{\rm cheb}-1 \\
    \end{cases} 
\end{align}
where $T_j$ is the standard $j$th Chebyshev polynomial.
\begin{align}
    T_j(s) := \cos (j \cos^{-1} s)
\end{align}
The node collocation is described by
\begin{align} \label{eq:cheb_node_def}
    s_k = \cos\left(\frac{2k-1}{2M_{\rm cheb}} \pi\right),\quad k \in \{1,2,\dots,M_{\rm cheb}\}.
\end{align}
These polynomials fulfill the discrete orthonormality condition \cite{mason2002chebyshev} with respect to the collocation nodes, that is,
\begin{align}
    \sum_{k=1}^{M_{\rm cheb}} u_i (s_k) u_j (s_k) = \delta_{ij}
\end{align}
for all $0\leq i, j < {M_{\rm cheb}}$.  With this, the interpolation is well conditioned as the condition number for $\mathbf{V}$ is optimal, i.e. $\kappa\left(\mathbf{V}\right)= \sigma_{\rm max} \left(\mathbf{V}\right)/\sigma_{\rm min} \left(\mathbf{V}\right)=1$. 

After simplifications, we can bypass the estimation of the coefficients $a_i$, which is a quadratic operation in ${M_{\rm cheb}}$, and use
\begin{align}
   P_{{M_{\rm cheb}}-1}f(0)=\sum^{M_{\rm cheb}}_{k=1} d_k f(s_k)
\end{align}
where
\begin{align}
    d_k &= \frac{1}{{M_{\rm cheb}}}(-1)^{k+{M_{\rm cheb}}/2} \tan\left(\frac{2k-1}{2{M_{\rm cheb}}}\pi\right).
\end{align}
We define the interpolation error to be 
\begin{align}
    \epsilon_{\rm cheb} = |f(0) - P_{{M_{\rm cheb}}-1}f(0)|.
\end{align}

The Trotter interpolation error bounds could be analyzed using the following Lemma. %, we need to introduce some methods of complex analysis. 
For each $\rho > 1$, let $B_\rho \subset \mathbb{C}$ be the Bernstein ellipse with foci at $\pm 1$ and semimajor axis $(\rho + \rho^{-1})/2$. The following lemma bounds the Chebyshev interpolation error for analytic functions on $B_\rho.$
\begin{lemma}\label{lem:Berns}
    Let $f(z)\in \mathbb{C}$ be an analytic function on $B_\rho$, and suppose $C\in\mathbb{R}_+$ is an upper bound such that $|f(z)| \le C$, for all $z \in B_\rho$. Then the Chebyshev interpolation error on $[-1,1]$ satisfies
    $$
    \norm{f - P_{{M_{\rm cheb}}-1} f}_{\infty} \leq \frac{4 C \rho^{-({M_{\rm cheb}}-1)}}{\rho - 1}
    $$
    for each degree $n>0$ of the interpolant through the ${M_{\rm cheb}}$ Chebyshev nodes.
\end{lemma}
The proof of this Lemma is given by Theorem 8.2 of Ref.~\cite{10.5555/3384673}.

In the following section we will explain how we use a new \cite{motlagh2024generalizedquantumsignalprocessing} Quantum Signal Processing method to use to implement the LWF for the function $e^{-\beta(x+1)}$ which is needed to implementing the Boltzmann factor $e^{-\beta (H+1)}$. 

\section{Quantum Signal Processing for Trotterized Evolution Operator}\label{sec:QSP_for_Trotterized_Ev_op}

One method of implementing the Trotterized evolution operator on a quantum circuit is using the Generalized Quantum Signal Processing (GQSP) \cite{motlagh2024generalizedquantumsignalprocessing}. The procedure consists of interleaving the following two operations, and the circuit consists of the system and an ancillary qubit. For an evolution operator $U$, the anti-controlled evolution operator is defined as:
\begin{equation}
A = (\vert 0\rangle\langle0\vert \otimes U)+(\vert 1\rangle\langle1\vert \otimes I) = \begin{bmatrix}

U & 0  \\

0 & I \\
\end{bmatrix},\end{equation}
and the arbitrary SU(2) rotations of the ancillary qubit:
\begin{equation}
R(\theta, \phi, \lambda) = \begin{bmatrix} e^{i(\lambda+\phi)}\cos(\theta) & e^{i\phi}\sin(\theta) \\
e^{i\lambda}\sin(\theta) & -\cos(\theta)  \\ \end{bmatrix} \otimes I.\label{eq:R}
\end{equation}
These operations allow the implementation of an arbitrary polynomial transformation:
\begin{align}
  \left|  P(U) \right| \leq 1
\end{align}
where $P\in \mathbb{C}[x] $ and $\deg(P)\leq d$. The main theorem formalizes these results in Ref.~\cite{motlagh2024generalizedquantumsignalprocessing}. %\textbf{cite Danial and Nathan}

\begin{theorem}[Generalized Quantum Signal Processing]\label{lem:PolynomialOfUnitaries}
$\forall d\in \mathbb{N},\>\exists\> \vec{\theta}, \vec{\phi} \in \mathbb{R}^{d+1},\>\lambda \in \mathbb{R}$ s.t:
\begin{equation}
    \left( \prod_{j=1}^{d} \mbox{$\large {R(\theta_{j}, \phi_{j}, 0) A}$} \right) R(\theta_0, \phi_0, \lambda) = 
    \begin{bmatrix}
        P(U) & .\>\>\>  \\
        
        Q(U) & .\>\>\> \\
    \end{bmatrix}\\
\end{equation}
\[
\mbox{\Large$\Longleftrightarrow$}
\]
\begin{center}
    \begin{varwidth}{\textwidth}
        \begin{enumerate}
            \item  $P, \,Q\in \mathbb{C}[x]$ and $\text{deg}(P), \text{deg}(Q) \leq d$.
            \item  $\forall x\in \mathbb{T},\> |P(x)|^2 + |Q(x)|^2 = 1$.
        \end{enumerate}
    \end{varwidth}
\end{center}
\end{theorem}
Here, $\mathbb{T}=\{x\in \mathbb{C} : |x|=1\}$.

Following this result, authors in \cite{motlagh2024generalizedquantumsignalprocessing} find that for a given $|P|^2 \leq 1$ on $\mathbb{T}$, there is a $Q$ to satisfy the requirements of Theorem.~\ref{lem:PolynomialOfUnitaries} and provided an efficient algorithm, $\tilde{O}({\rm poly}(d))$, to determine the necessary angles $\theta_j$ and $\phi_j$ for a given $P(U)$.

\begin{comment}
\begin{lemma}\label{lem:QExistence}
 $\forall P\in \mathbb{C}[x]$, we have:
\begin{equation}
\forall x\in \mathbb{T}, \>|P(x)|^2 \leq 1
\end{equation}
\[
\mbox{\Large$\Longleftrightarrow$}
\]
\begin{equation}
\exists Q\in \mathbb{C}[x]\>\> s.t. \>\>\textrm{deg}(P)=\textrm{deg}(Q) \>\> \wedge \>\> \forall x\in \mathbb{T},\> |P(x)|^2 + |Q(x)|^2 = 1
\end{equation}
\end{lemma}
\end{comment}

This leads to the following corollary, which states that we can build any arbitrary (appropriately scaled) polynomial P of U using our technique
\begin{corollary}\label{cor:PossiblePs}
 $\forall P\in \mathbb{C}[x]$, with $\text{deg}(P) = d$ if:
\begin{equation}
\forall x\in \mathbb{T}, \>|P(x)|^2 \leq 1
\end{equation}
Then $\exists\> \vec{\theta}, \vec{\phi} \in \mathbb{R}^{d+1}, \> \exists \lambda \in \mathbb{R}$ such that:
\begin{equation}\label{eq:phases_condition} 
\begin{bmatrix}
P(U) & .\>\>\>  \\
.\>\>\> & .\>\>\> \\
\end{bmatrix} = \left( \prod_{j=1}^{d} \mbox{$\large {R(\theta_{j}, \phi_{j}, 0) A}$} \right) R(\theta_0, \phi_0, \lambda)
\end{equation}
\end{corollary}
Such a result shows that the GQSP advances the traditional QSP method in the sense that it provides a simplification for the computation of phase angles, which is a challenge for the traditional QSP method. Further, it does not require the fixed parity of the function, which is another caveat of the traditional QSP method. 

\onecolumngrid

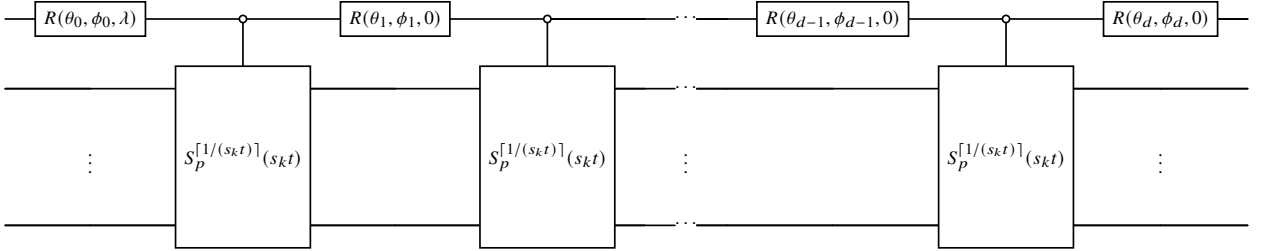
\begin{figure}[h!]
\centering

% \[
% \begin{quantikz}
% & \gate{R(\theta_0, \phi_0, \lambda)} & \octrl{1} & \gate{R(\theta_1, \phi_1, 0)} & \octrl{1} & \qw & \cdots & \qw & \gate{R(\theta_{d-1}, \phi_{d-1}, 0)} & \octrl{1} & \gate{R(\theta_d, \phi_d, 0)} & \qw \\
% & \qw & \gate[3]{U} & \qw & \gate[3]{U} & \qw & \cdots & \qw & \qw & \gate[3]{U} & \qw & \qw \\
% & \vdots &  &  & &  & \cdots &  &  &  & &  \\
% & \qw & \qw & \qw & \qw & \qw & \cdots & \qw & \qw & \qw & \qw & \qw
% \end{quantikz}
% \]
\scalebox{0.8}{
\begin{quantikz}
& \gate{R(\theta_0, \phi_0, \lambda)} & \octrl{1} & \gate{R(\theta_1, \phi_1, 0)} & \octrl{1} & \qw & \cdots & \qw & \gate{R(\theta_{d-1}, \phi_{d-1}, 0)} & \octrl{1} & \gate{R(\theta_d, \phi_d, 0)} & \qw \\
& \qw & \gate[3]{S^{\lceil 1/(s_k t) \rceil}_p(s_k t)} & \qw & \gate[3]{S^{\lceil 1/(s_k t) \rceil}_p(s_k t)} & \qw & \cdots & \qw & \qw & \gate[3]{S^{\lceil 1/(s_k t) \rceil}_p(s_k t)} & \qw & \qw \\
\setwiretype{n} &\vdots &  & &&& \vdots &  & & & \vdots & \\% & \vdots & & & & & & & & & & \\
& \qw & \qw & \qw & \qw & \qw & \cdots & \qw & \qw & \qw & \qw & \qw
\end{quantikz}
}
\caption{Quantum circuit for GQSP wherein the angles $\theta$ and $\phi$ are chosen to enact a polynomial transformation of $U$\label{fig:GQSP_circuit}.  Here $R(\theta,\phi,\lambda)$ is a general single qubit rotation as specified in~\eqref{eq:R}.}
\end{figure}

\twocolumngrid
In this work, we consider the Unitary operator to be the Trotterized operator as $U = S_p^{\lceil 1/s_kt \rceil}(s_k t)$ and the polynomial to be the low-weight Fourier expansion of a function as shown in Ref.~\cite{van_Apeldoorn_2020}. 
\begin{equation}
\begin{aligned}
    f(x) & \approx \sum_{k=0}^K a_k \sum_{l = 0}^{L} b_l^{k} \Big(\frac{i}{2}\Big)^l \sum_{m = \frac{l}{2} - M'}^{\frac{l}{2} + M'}(-1)^m \binom{l}{m} e^{\frac{i\pi x (2m-l)}{2}}\\
    & = \sum_{m = -M}^M c_m e^{\frac{i\pi x m}{2}}
    \label{eq:fourier_mapprox}
\end{aligned}
\end{equation}
where $a_k$ are the kth order Taylor coefficients of the function $f(x)$ and $b_l^k$ is defined as $\Big(\frac{\arcsin(y)}{\pi/2}\Big)^k = \sum_{l=0}^{\infty} b_l^k y^l$. The circuit of the GQSP implementation of the Trotterized Unitary is shown below in Fig.~\ref{fig:GQSP_circuit}.

% \begin{figure}
% \[
% \begin{array}{rcl}
% \text{(a)} & 
% \begin{quantikz}
% \qw&\gate{\hat{V}(\theta)} & \qw
% \end{quantikz}
% \quad &= \quad
% \begin{quantikz}
% \qw&\gate{\hat{R}_{\phi_1}(\theta)} & \gate{\hat{R}_{\phi_2}(\theta)} & \qw & \dots & \qw & \gate{\hat{R}_{\phi_N}(\theta)} & \qw
% \end{quantikz} \\[1em]
% \text{(b)} & 
% \begin{quantikz}
% \qw&\gate[2]{\hat{U}_{\phi}} & \qw \\
% \qw&\qw & \qw
% \end{quantikz}
% \quad &= \quad
% \begin{quantikz}
% \qw&\gate{e^{i\phi \hat{\sigma}_z / 2}} & \gate{H} & \ctrl{1} & \gate{H} & \gate{e^{-i\phi \hat{\sigma}_z / 2}} & \qw \\
% \qw&\qw & \qw & \gate{\hat{S}_p^{\lfloor 1/(s_k t)\rfloor}(s_k t)} & \qw & \qw & \qw
% \end{quantikz} \\[1em]
% \text{(c)} & 
% \begin{quantikz}
% \qw & \gate[2]{\hat{U}_{\phi}} & \qw \\
% \lstick{$|u_\lambda\rangle$} & \qw & \qw
% \end{quantikz}
% \quad &= \quad
% \begin{quantikz}
% \qw & \gate{\hat{R}_{\phi}(\theta_\lambda)} & \qw \\
% \lstick{$|u_\lambda\rangle$} & \qw & \qw
% \end{quantikz} \\[1em]
% \text{(d)} & 
% \begin{quantikz}
% \qw & \gate[2]{\hat{V}} & \qw \\
% \lstick{$|\psi\rangle$} & \qw & \qw
% \end{quantikz}
% \quad &= \quad
% \begin{quantikz}
% \qw & \gate[2]{\hat{U}_{\phi_1}} & \gate[2]{\hat{U}_{\phi_2}} & \qw & \dots & \qw & \gate[2]{\hat{U}_{\phi_N}} & \qw \\
% \lstick{$|\psi\rangle$} & \qw & \qw & \qw & \dots & \qw & \qw & \qw
% \end{quantikz}
% \end{array}
% \]

%     \caption{Quantum signal processing Circuit}
%      \label{fig:qsp}
% \end{figure}

\section{Cost analysis} \label{sec:cost_analysis}

First, we review the method of preparing the Gibbs state with the thermofield double state and quantum amplitude estimation (QAE), and further, we analyze the Trotter Interpolation Bounds and show that we could achieve the state-of-the-art cost efficiently using Trotterization and interpolation. 

\subsection{Thermofield Double State }

We review a known way of preparing the Gibbs state using quantum amplitude estimation following Ref.~\cite{Poulin2009}. 
The Gibbs state could be prepared from a thermal field double state by partial trace of part of the system. 
Therefore, the key to the preparation procedure is to prepare the thermal field in a double state at finite temperature.
An infinite temperature thermofield double state for systems $A,B$ with dimension $N$ could be written as 
\begin{align}
    \ket{\psi_0}_{AB} = \frac{1}{\sqrt{N}} \sum^{N-1}_{n=0} \ket{n}_{A} \otimes \ket{n}_{B}
\end{align}
For the following effective Hamiltonian
\begin{align}
   \tilde{H}_p =  \frac{ \log{S_{p} \left(s t\right)} } {i st}
\end{align}
We define the following operator
\begin{align}
    U_{{\rm boltz},s} =  e^{-\beta (\tilde{H}_p+1)/2} \otimes \ket{0}\bra{0}_C  +   G  ,
\end{align}
with an ancillary qubit register C. Here, $e^{-\beta (\tilde{H}_p +1)/2}$ acts either on register $A$ or $B$, any of those works. The operator $G$ fulfills the property $\left(I_{
AB} \otimes \bra{0}_C\right) G \left( I_{AB}\otimes \ket{0}_C \right)=0$.

\begin{lemma}
The unitary operator
\begin{align*}
    U_{{\rm boltz},{s_k}} =  e^{-\beta (\tilde{H}_p(s_k)+1)/2} \otimes \ket{0}\bra{0}  +   G  ,
\end{align*}
with a garbage operator $G$ with the property $\left(I \otimes \bra{0}\right) G \left( I\otimes \ket{0} \right)=0$ can be implemented with an error $\epsilon_{QSP}$ using $\tilde{O}(\beta \log(\epsilon_{QSP}))$ applications of  $S^{\lceil 1/(s_k t) \rceil}_{p} \left( s_k t\right)$.
\end{lemma}
\begin{proof}
This is done through Generalized Quantum Signal Processing, where the circuit implements an approximation of the following operator, a function of the unitary operator $S^{\lceil 1/(s_k t) \rceil}_{p} \left( s_k t\right)$ (for positive $s_k$, for negative $s_k$ we use the floor function):
    \begin{align}
        f\left(S^{\lceil 1/(s_k t) \rceil}_{p}\right) &= \exp\left( i \beta_k \log\left(S^{\lceil 1/(s_k t) \rceil}_{p} \left( s_k t\right)\right)\right) \cr
        &= \exp\left( - \beta_k \left(s_k t\lceil 1/ (s_k t) \rceil\right) \tilde{H}_p (s_k t)\right)
    \end{align}
    Here, $\beta_k = \beta \frac{\lceil 1/ s_k t \rceil }{1/ s_k t}$. With this correction factor we avoid using fractional queries but obtain the right inverse temperature with respect to the effective Hamiltonian $\tilde{H}_p$. Moreover, it has been shown in \Cref{cor:M_k} that the degree order scales like $O(\beta_k \log{\epsilon_{QSP}})$ or equivalently $O(\beta \log{\epsilon_{QSP}})$.
    
\end{proof}

We note that measuring zero on the register $C$ using the following state
\begin{align}
   U_{{\rm boltz},s} \left( \ket{\psi_0}_{AB}\otimes \ket{0}_C\right)
\end{align}
has the probability of:
\begin{align}
    p_0 = \frac{\sum_n e^{-\beta \tilde{\lambda}_{n}}}{N} = \frac{Z(\beta,s)}{N}.
\end{align}
We can estimate the corresponding amplitude following Ref.~\cite{Brassard_2002}
\begin{align}
    a_0 = \sqrt{p_0}
\end{align}
through the following Grover operator
\begin{align}
    Q = - \mathcal{A} S_0 \mathcal{A}^{-1} S_{\chi}
\end{align}
where
\begin{align}
    \mathcal{A}  \left(\ket{0}_{AB} \otimes \ket{0}_C\right) = U_{{\rm boltz},s} \left(\ket{\psi_0}_{AB} \otimes \ket{0}_C \right),
\end{align}
\begin{align}
    S_0 = e^{i \pi \ket{0}\bra{0}_{ABC}} = I - 2\ket{0}\bra{0}_{ABC} 
\end{align}
and
\begin{align}
    S_\chi = e^{i \pi \ket{0}\bra{0}_{C}} = I - 2\ket{0}\bra{0}_{C} .
\end{align}
Thus, we can estimate $a_0$ with an additive error tolerance of $\varepsilon$ using $O(1/\varepsilon)$ controlled applications of $Q$. The uncertainty propagated to $p_0$ is then 
\begin{align}
    O\left(\sqrt{\frac{Z(\beta,s)}{N}} \varepsilon\right).
\end{align}
Here, we have illustrated one of the known ways to get the quadratic speed up for Gibbs state preparation. This is summarized in the following lemma:

\begin{lemma}\label{lem:ae_Nboltz}
There exists a quantum algorithm to estimate the canonical partition function $Z(\beta,s)$ of a Trotterized Hamiltonian, $\tilde{H}_p =  \frac{ \log{S_{p} \left(s t\right)} } {i st}$, with a statistical error $\varepsilon$ and a number of
\begin{align*}
    O\left(\frac{\sqrt{Z(\beta,s)/N}}{\varepsilon}\right)
\end{align*}
$U_{{\rm boltz},s}$ queries using $2\log_2 N+2$ qubits.
\end{lemma}
The extra $\log_2 {N}$ qubits are for the trace operation, one extra qubit is for the generalized signal processing to apply $e^{- \beta \tilde{H}_p(s)}$, and the other is for quantum amplitude estimation algorithms (See \cite{grinko2021iterative}, for example).

\subsection{Trotter Interpolation Bounds}

%\subsection{Chebyshev interpolation}

Starting from the Trotterization error, from Lemma 12 in Ref.~\cite{rendon2023dequantizing}, the difference between the Hamiltonian and the effect Hamiltonian scales as 
\begin{align}
   \| L_{p}(\tau)\|=\|\tilde{H}_p (\tau) - H \| = O\left( \frac{|\tau|^p \alpha}{(p+1)!} \right),
\end{align}
for a small enough $|\tau|$, with $\tau = z t$, and $z$ is the complex continuation of the real variable $s$, and $\alpha \in \mathbb{R}^+$ be constant. 

Following this result, we can consider the operator
\begin{align}
    e^{-\Delta_{\beta} \tilde{H}_p (\tau)}
\end{align}
for some sufficiently small real parameter $\Delta_\beta$ such that we can use the Zassenhaus formula:
\begin{equation}
        e^{t(X+Y)} = e^{tX}~ e^{tY} ~e^{-\frac{t^2}{2} [X,Y]} ~
e^{\frac{t^3}{6}(2[Y,[X,Y]]+ [X,[X,Y]] )} ~
%e^{\frac{-t^4}{24}([[[X,Y],X],X] + 3[[[X,Y],X],Y] %+ 3[[[X,Y],Y],Y]) } 
\cdots.
\end{equation}
which leads to,
\begin{align}
    e^{-\Delta_\beta \tilde{H}_p (\tau)} = e^{-\Delta_\beta H} e^{-\Delta_\beta L_p (\tau)} e^{-\frac{\Delta_\beta^2}{2}[H,L_p(\tau)]} \dots
\end{align}
With this, we want to estimate the error on the trace
\begin{align}
    \Tr \left( e^{-\beta \tilde{H}_p}\right)
\end{align}
with respect to $Z=\Tr \left(e^{-\beta H}\right)$. First we decompose $L_p(\tau)$ into the eigenbasis of $H$ and note that 
\begin{align}
  \left| \frac{\Tr\left(e^{-\beta \tilde{H}_p}\right)}{N} \right|& \leq \frac{1}{N}\sum_{\lambda}\left( e^{-\beta E_\lambda} \left(e^{\beta \| L_p (\tau) \|}\right)\right) \cr 
   &= O\left(e^{\frac{\beta |\tau|^{p}}{p!} \alpha}\right) \frac{Z}{N}
\end{align}
If we consider the disc $\tau = rt$, with radius $r$, enclosing such a Bernstein ellipse $B_{\rho}$ we can bound the trace within this disc with

\begin{align}
    \max_{z \in B_{\rho}}\left|\frac{\Tr\left(e^{-\beta \tilde{H}_p}\right)}{N}\right|& = O\left(e^{\frac{\beta (r t)^{p} \alpha }{p!} }\right) \frac{Z}{N}.
\end{align}
Here, the relationship between $\rho$ and $r$ is
\begin{align} 
  \rho = r + \sqrt{r^2-1}.
\end{align}
Now, using Lemma.~\ref{lem:Berns}, one knows the error from nth order Chebyshev interpolation scales as 
\begin{align}
    \epsilon_{\rm cheb} &= O\left(\frac{e^{ \alpha \beta (r t)^p/p!} \left(Z/N\right)}{r^{M_{\rm cheb}}}\right) \cr 
    &= O\left(\frac{e^{ \alpha \beta (r t)^p}/p!}{e^{M_{\rm cheb}\log{r}}}\right).
\end{align}
Solving for the order $n$, one gets
\begin{align}
   M_{\rm cheb}=\frac{1}{\log{r}} O\left(\log\left(\frac{Z/N}{\epsilon_{\rm cheb}}\right)   + \frac{\beta\alpha (r t )^p}{p!}\right).
\end{align}

We now choose the following scalings:
\begin{align}
    p \sim \sqrt{\log_5 \beta} \cr 
    t \sim e^{- \sqrt{\log \beta \log 5}} \cr
    r \sim e^{1/\sqrt{\log_5 \beta}}.
\end{align}
With this
\begin{align}
   M_{\rm cheb}=O\left( \sqrt{\log_5{\beta}}  \left( \log\left(\frac{Z/N}{\epsilon_{\rm cheb}}\right)   + \frac{\alpha }{(\sqrt{\log_5 \beta}/e)^{\sqrt{\log_5 \beta}}}\right)\right).
\end{align}
\begin{lemma}\label{lem:Mcheb}

One can estimate 
\begin{align}
    \frac{Z(\beta)}{N} = \frac{{\rm Tr}\left(e^{-\beta H}\right)}{N}
\end{align}
for a Hamiltonian $H=\sum^{\Gamma}_\gamma H_\gamma$ and $\| H \| \leq 1$ with an \textbf{algorithmic error} $\epsilon_{\rm cheb}$ sampling
\begin{align}
    \frac{Z(\beta,s)}{N} = \frac{{\rm Tr}\left(e^{-\beta \tilde{H}_p (s t) }\right)}{N}
\end{align}
at $ s_k= \cos\left(\frac{2k-1}{2M_{\rm cheb}} \pi\right)$, where $k \in \{1,2,\dots,M_{\rm cheb}\}$ and using the following linear combination:
\begin{align}
   Z(\beta,0) \approx \sum^{M_{\rm cheb}}_{k=1} d_k Z(\beta,s_k),
\end{align}
where
\begin{align}
    d_k &= \frac{1}{M_{\rm cheb}}(-1)^{k+M_{\rm cheb}/2} \tan\left(\frac{2k-1}{2M_{\rm cheb}}\pi\right),
\end{align}
and
\begin{align}
   M_{\rm cheb}=\tilde{O}\left(  \log\left(\frac{Z/N}{\epsilon_{\rm cheb}}\right)   + {\alpha }\right).
\end{align} 
\end{lemma}

The depth of the circuit to implement each $e^{-\beta \tilde{H}_p(t s_k)}$ (or $U_{{\rm boltz},s_k}$) scales as :
\begin{align}
    O\left(M_k \frac{5^p}{t |s_k|} \right),
\end{align}
where $M_k$ is the order of the polynomial to implement the operator $e^{-\beta \tilde{H}_p (ts_k)}$, the $5^p$ comes from the stages of the Suzuki-Trotter formula. Thus, in order to estimate each $Z(\beta,s_k)/N$ with a statistical uncertainty $\tilde{O}(\varepsilon)$ the quantum computational cost is (See ~\Cref{lem:ae_Nboltz}) :
\begin{align}
     O\left(M_k \frac{5^p}{t |s_k|} \frac{\sqrt{Z(\beta,s_k)}/N}{\varepsilon} \right).
\end{align}
Summing the cost for all $Z(\beta,s_k)/N$:
\begin{align}
 O\left(\sum^{M_{\rm cheb}}_{k=1}M_k \frac{5^p}{t |s_k|} \frac{\sqrt{Z(\beta,s_k)}/N}{\varepsilon} \right)
\end{align}
The total cost, then, is expected to scale as 
\begin{align}
    O\left(\frac{5^p}{t}\max_k \left({M_k} \frac{\sqrt{Z(\beta,s_k)}/N}{\varepsilon} \right)  M_{\rm cheb} \log M_{\rm cheb}\right).
\end{align}
where we used $\sum_k 1/|s_k| \sim O(M_{\rm cheb}\log M_{\rm cheb } ) $ as shown in \cite{Rendon2024improvedaccuracy}. 
From Lemma.~\ref{lemma:LowWeightAPX}, the order of the low-weight Fourier expansion scales as 
\begin{align}
    M_k = O(\beta \log{1/\epsilon_{\rm QSP}}),
\end{align}
where $\epsilon_{\rm QSP}$ is the target operator error for the implementation of $U_{{\rm boltz},{s_k}}$(which encodes $e^{-\beta \tilde{H}_p (ts_k)}$). With this, the cost scales as 
\begin{align}
    \tilde{O}\left(\frac{\beta 5^p}{t} \log\left(\frac{1}{\epsilon_{\rm QSP}}\right) M_{\rm cheb}\log{M_{\rm cheb }}\left(\frac{\sqrt{\max_{s} Z(\beta,s)/N}}{\varepsilon}\right)\right).
\end{align}
Now, choosing again like in \Cref{lem:Mcheb} the following scalings:
\begin{align}
    p \sim \sqrt{\log_5 \beta}, \cr 
    t \sim e^{- \sqrt{\log \beta \log 5}};
\end{align}
and using $M_{\rm cheb}=\tilde{O}\left(  \log\left(\frac{Z/N}{\epsilon_{\rm cheb}}\right)   + {\alpha }\right)$ from that same lemma, the cost scales in the following way
\begin{align}
    \tilde{O}\left(\beta e^{2 \sqrt{\log\beta \log 5}}\log\left(\frac{1}{\epsilon_{\rm QSP}}\right)  \left(\log\left( \frac{Z/N}{\epsilon_{\rm cheb}}\right)   + \alpha \right)\right.\cr 
    \left.\times \left(\frac{\sqrt{\max_{s} Z(\beta,s)/N}}{\varepsilon}\right)\right).
\end{align}
Ignoring the sub-polynomial factor $e^{2\sqrt{\log \beta \log 5}}$, which asymptotically grows slower than $\beta^a$ for any $a >0$, we see that the expected behavior is to all practical purposes linear in $\beta$. Moreover, we will set the errors $\epsilon_{QSP} \sim \epsilon_{\rm cheb}$, and given that $Z/N$ is expected to be very small we can ignore the scaling of $\log\left(\frac{Z/N}{\epsilon_{\rm cheb}}\right)$ as it scales logarithmically with the relative error and not the additive error. Finally, we can consider the cost being
\begin{align}
     \tilde{O}\left(\beta\max({\alpha,1}) \log{1/\epsilon}\left(\frac{\sqrt{\max_{s} Z(\beta,s)/N}}{\varepsilon}\right)\right).
\end{align}
Finally, our main results are condensed in the following theorem:
\begin{theorem}
    One can estimate 
\begin{align}
    \frac{Z(\beta)}{N} = \frac{{\rm Tr}\left(e^{-\beta H}\right)}{N}
\end{align}

for a Hamiltonian $H=\sum^{\Gamma}_\gamma H_\gamma$ for fast-forwardable $H_{\gamma}$ and $\| H \| \leq 1$ with \textbf{statistical error} $\varepsilon$ and \textbf{algorithmic error} $\epsilon$ by first estimating
\begin{align}
    \frac{Z(\beta,s_k)}{N} = \frac{{\rm Tr}\left(e^{-\beta \tilde{H}_p (s_k t) }\right)}{N}
\end{align}
for $ s_k= \cos\left(\frac{2k-1}{2M_{\rm cheb}} \pi\right)$, where $k \in \{1,2,\dots,M_{\rm cheb}\}$, each with a with \textbf{statistical error} $\tilde{O}(\varepsilon)$ and \textbf{algorithmic error} $\tilde{O}(\epsilon)$. Here $\tilde{H}_p(s t)$ is the effective Hamiltonian to the Trotterized evolution by a $p$th-order formula $S_{p} ( s t)$. With these, said estimate is obtained through the following linear combination:
\begin{align}
   Z(\beta,0) \approx \sum^{M_{\rm cheb}}_{k=1} d_k Z(\beta,s_k),
\end{align}
where
\begin{align}
    d_k &= \frac{1}{M_{\rm cheb}}(-1)^{k+M_{\rm cheb}/2} \tan\left(\frac{2k-1}{2M_{\rm cheb}}\pi\right).
\end{align}
This, with a cost
\begin{align*}
     \tilde{O}\left(\beta \max\left(\alpha,1\right) \log{1/\epsilon}\left(\frac{\sqrt{\max_{s} Z(\beta,s)/N}}{\varepsilon}\right)\right)
\end{align*}
provided we set the Trotter-step interval to have scale of $t \sim e^{-\sqrt{\log\beta \log 5}}$ and the product formula order to scale like  $p \sim \sqrt{\log_5 \beta}$.
\end{theorem}

\section{Numerical Results} \label{sec:numerical_results}
We consider the simulation as shown in Figure.~\ref{fig:GQSP_circuit} with the SYK model
\begin{comment}
\begin{equation}
    H = \sum_{i} -t (c_i^{\dagger}c_{i+1} + c_{i+1}^{\dagger}c_i) + U n_i n_i
\end{equation}
This model can be mapped to the XXZ model as 
\begin{equation}
    H_{xxz} = \sum_i -J(S_i^xS_{i+1}^x+S_i^yS_{i+1}^y+\Delta S_i^zS_{i+1}^z)
\end{equation} 
\end{comment}
\begin{equation}
    H_{xxz} =  \frac{1}{4*4!} \sum^{n_{\rm Majorana}}_{i,j,k,l = 1} J_{ijkl} \gamma_i \gamma_j \gamma_k \gamma_l
\end{equation} 
where $\gamma_i$ are Majorana operators with anti-commutator relation $\{\gamma_i, \gamma_j\} = \delta_{ij}$. The model could be mapped to a spin model on $n_{\rm Majorana}/2$ qubits with the standard Jordan-Wigner transformation \cite{Luo_2019, Garc_a_lvarez_2017, Asaduzzaman_2024, Babbush_2019}
\begin{align}
    \gamma_{2k-1} = \frac{1}{\sqrt{2}}\Bigg(\prod_{j=1}^{k-1}Z_j\Bigg) X_k \,\,\mathbb{I}^{\bigotimes (n_{\rm Majorana}-2k)/2} \\
        \gamma_{2k} = \frac{1}{\sqrt{2}}\Bigg(\prod_{j=1}^{k-1}Z_j\Bigg) Y_k \,\,\mathbb{I}^{\bigotimes (n_{\rm Majorana}-2k)/2} 
\end{align}
The straightforward operator counts show that the number of Pauli strings and the stages scale as $\mathcal{O}(n_{\rm Majorana}^4)$ from the number of terms of $n_{\rm Majorana} \choose 4$. The simulation cost could be reduced by regrouping the commuting terms into the forms of \cite{Asaduzzaman_2024}
\begin{equation}
    H = \sum_i H_i,  \quad [H_i, H_j]\neq 0 \quad \text{for} \quad i\neq q
\end{equation}
with 
\begin{equation}
    H_i = \sum_j \alpha_j \prod_k \sigma_k^j,  \quad \Big[\prod_k \sigma_k^j, \prod_l \sigma_l^m\Big] = 0 \quad \forall m,j
\end{equation}

Moreover, as a numerical confirmation of the convergence rates from before, we obtain the errors for a lower-weight Fourier expansion and the Taylor expansion of the function $f(H) = e^{-\beta(H+1)}$. 
For the LWF expansion, the target polynomial for the Trotterized operator $S_p^{\lceil 1/s_kt \rceil}(s_k t)$ could be written as $P = \sum c_m S_p^{m\lceil 1/s_kt \rceil}(s_k t)$ as shown in Eq.~\ref{eq:fourier_mapprox}. 
With the polynomials defined, one can implement the simulation as shown above in Fig.~\ref{fig:GQSP_circuit}, and the results of the error from purely truncating the polynomials are shown in Fig.\ref{fig:gqsp_result}.  To find the phases for \ref{fig:GQSP_circuit} one just needs to numerically solve for them in  \cref{eq:phases_condition}. In \cite{motlagh2024generalizedquantumsignalprocessing}
authors provide an equation and an efficient package to solve for these phases given the polynomial coefficients.

Our results show that with the GQSP procedure, the Gibbs preparation could be simulated as expected with controlled accuracy. We linearly fit the numerical results and found that the order $M$ of expansion scales linearly with $M\sim\log(1/\epsilon)$. 
\begin{figure}[h!]
    \centering    
    \includegraphics[width=0.49\textwidth]{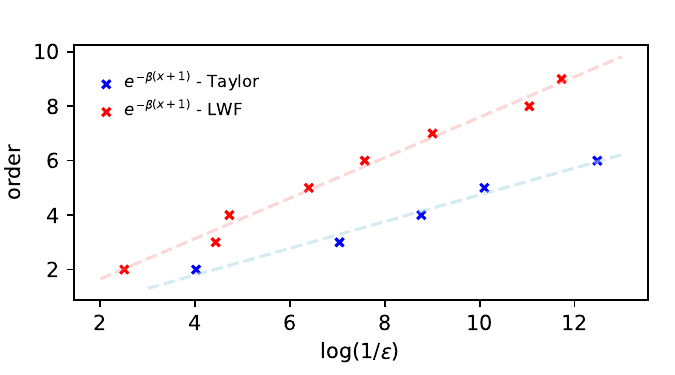}
\caption{GQSP of the function $e^{-\beta(x+1)}$ implemented using Taylor expansion and the low-weight Fourier expansion in Eq.~\ref{eq:fourier_mapprox}. The dashed lines are linear fits. }
    \label{fig:gqsp_result}
\end{figure}
The gates cost scales linearly with the order of the polynomial. Therefore, we observed that the cost is proportional to $\tilde{O}\left(\log{1/\epsilon}\right).$ as predicted from analytical results. 

Further, we show the results of the qubits saved using the interpolation method compared with the quantum walk operator method. As shown in Ref.~\cite{Low_2019}, the ancillary qubits needs for the qubitization is in the order of $O(\log_2(\Gamma))$ where $\Gamma$ is the number of stages for the Hamiltonian scales as $O(n_{\rm Majorana}^4)$. In comparison, our method of interpolation only requires 1 ancillary qubit using the GQSP. Therefore, we save the qubits extensively in the order of $O(\log_2(\Gamma))$, and it is shown in Fig~\ref{fig:qubits_saved}.
\begin{figure}[h!]
    \centering    
    \includegraphics[width=0.49\textwidth]{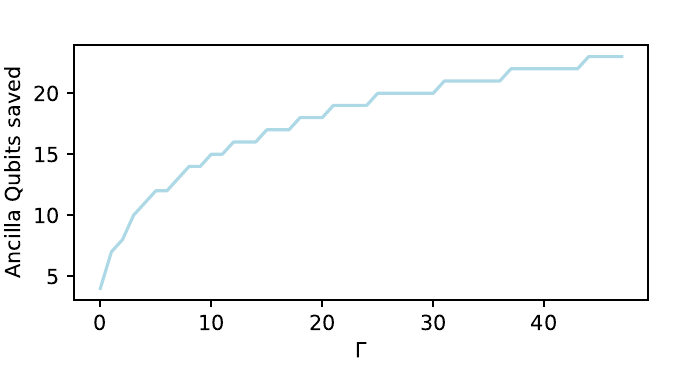}
\caption{ancillary qubits saved compared with the qubitization method from Ref.~\cite{Low_2019} or other quantum signal processing methods that rely on a quantum walk operator/block-encoding~\cite{Gily_n_2019}}
    \label{fig:qubits_saved}
\end{figure}

\section{Conclusions} \label{sec:conclusions}
%\textcolor{blue}{TO DO: Include a high-level (non-mathematical) implication of the conclusions.}

With the methods presented here, we have reduced the number of qubits required for Gibbs state preparation while preserving the asymptotic scaling, $\tilde{O}(\beta \log{1/\epsilon})$, of state-of-the-art methods.
The numerical implementation of Gibbs state preparation for the SYK model was benchmarked and the results fit the analytical derivation. 

We have used a new convergence rate estimation for the Trotter interpolation using Chebyshev polynomials/nodes using the trace instead of the spectral norm, which gets us a $\log{1/\epsilon}$ scaling instead of $\log^2 1/\epsilon$. 
Moreover, the scalings chosen for $t$, the order of the product formula $p$, and the bounding radius $r$, have allowed us to obtain a quasi-linear scaling with $\beta$.

We have swapped the quantum walk operator for the Trotterized evolution operator and performed the quantum eigenvalue transformation on the effective Hamiltonian $\tilde{H}_{p}$, instead of the typical transformed eigenphase $\arccos{(H)}$. 
After that, Chebyshev interpolation was used to extrapolate the partition function for the effective Hamiltonian to the partition function at Trotter-step size zero. With the bypassing of the transformation on the quantum walk operator, we have saved $\log{\Gamma}$ qubits as well as avoided the more complicated controlled operations required by block-encoding.

In future work, we would hope to have more detailed resource estimates in terms of gate counts and comparisons to other state-of-the-art methods. We would also like to extend the resource estimates to other types of systems with relatively fast growth of stages $\Gamma$ which benefit the most from this method.

\begin{acknowledgments}

We would like to thank Dr. Sarvagya Upadhyay of Fujitsu Research of America, Inc., for his valuable feedback on the manuscript. We are also deeply grateful to Dr. Hirotaka Oshima and Dr. Yasuhiro Endo of Fujitsu Research of Japan, Inc., for their insightful comments and suggestions throughout the course of the project.

\end{acknowledgments}
\bibliography{refs}

\end{document}